\documentclass[journal,doublecolumn,10pt]{IEEEtran}
\usepackage{epsfig,latexsym}
\usepackage{float}
 \usepackage{enumerate}
\usepackage{indentfirst}
\usepackage{amsmath}
\usepackage{amssymb}
\usepackage{times}
\usepackage{subfigure}
\usepackage{geometry}
\usepackage{cite}
\usepackage{url}
\usepackage{color}
\definecolor{green}{rgb}{0.796,0.948,0.816}
\newtheorem{Remark}{Remark}

\newtheorem{theorem}{$\mathbf{Theorem}$}
\newtheorem{lemma}[theorem]{$\mathbf{Lemma}$}

\newtheorem{corollary}[theorem]{$\mathbf{Corollary}$}

\begin{document}
\title{The Private Key Capacity of a Cooperative Pairwise-Independent Network}
\author{\IEEEauthorblockN{Peng Xu\IEEEauthorrefmark{1},
  Zhiguo Ding\IEEEauthorrefmark{2} and Xuchu Dai\IEEEauthorrefmark{1}\\}
\IEEEauthorblockA{\IEEEauthorrefmark{1}Dept. of Elec. Eng. and Info Sci., Univ. Science \& Tech. China, Hefei, Anhui, 230027, P.R.China\\}
\IEEEauthorblockA{\IEEEauthorrefmark{2}School of
Computing and Communications, Lancaster University, LA1 4WA, UK\\}
\vspace{-2em}}
 \maketitle

\begin{abstract}
This paper studies the private key generation of a cooperative pairwise-independent network (PIN) with
$M+2$ terminals (Alice, Bob and $M$ relays), $M\geq 2$. In this PIN, the  correlated sources observed by
every pair of terminals are independent of those sources observed by any other pair of terminal.
All the terminals can communicate with each other over a public channel which is also observed by Eve noiselessly.
The objective is to generate a private key between Alice and Bob under the help of the $M$ relays;
such a private key needs to be protected not only from Eve but also from individual relays simultaneously.
The private key capacity of this PIN model is established,
whose lower bound is obtained by proposing a novel random binning (RB) based key generation algorithm,
and the upper bound is obtained based on the construction of $M$ enhanced source models.
The two bounds are shown to be exactly the same. Then, we consider a cooperative wireless network
and use the estimates of fading channels to generate private keys. It has been shown that the
proposed RB-based algorithm can achieve a multiplexing gain $M-1$, an improvement in comparison with
the existing XOR-based algorithm whose achievable  multiplexing gain is $\lfloor M \rfloor/2$.

%{\bf non-colluding}
%$\lfloor \cdot \rfloor$ denotes the floor operation.
\end{abstract}\vspace{-0.2em}
\vspace{-1.5em}
\section{Introduction}\label{i}
 The pairwise-independent network (PIN) was introduced in \cite{ye2007group} for secret key generation.
 Since then, many other related works have also investigated a variety of PIN models  (e.g.,
 \cite{nitinawarat2010secret,lai2012simultaneously,lai2013simultaneously}),
%\cite{csiszar2000common,csiszar2004secrecy,ye2005secret,ye2007group,nitinawarat2010secret,lai2013simultaneously}
and each of them aimed to find the secret key capacity  of a particular PIN model.
 The PIN model is actually a special case of the multi-terminal ``source model"  \cite{csiszar2000common,csiszar2004secrecy},
  in which the  correlated sources observed by
every pair of terminals are independent of those sources observed by any other pair of terminal.
Note that the so-call ``source model'' was first studied by Ahlswede and Csis\'{a}r for generating secret keys between two terminals
using their correlative observations and  public transmissions \cite{ahlswede1993common}.

In recent years, the PIN model has been applied to practical wireless communication networks for key generation.
Based on channel reciprocity, the correlated source observations in a PIN model can be obtained via estimating the wireless fading channels
associated with legitimate terminals. This is because all the wireless channels in a network are mutually independent as long as
the terminals are half-wavelength away from each other  \cite{tse2005fundamentals}.
This physical layer (PHY) security  approach has been recognized as a promising solution %liu2012exploiting
 for generating secret key in  recent years (e.g.,\cite{ye2010information,wang2012cooperative,
 lai2012cooperative,zhou2014secret}).
%  {In addition to these source-model-based secure methods,  there also exists another type of
%   works  in the area
%    of PHY security,  which is based on  channel models. Compared to the secrecy communications
%in  channel models (e.g., \cite{Wyner1975,csisz¨¢r1978broadcast,lai2008relay,xu2013rate,xu2014general}), the PHY key generation approach
%  in source models \cite{ye2010information,lai2012unified,quist2013optimal,wang2011fast,
% wang2012cooperative,lai2012cooperative,zhou2014secret,chen2013smoke}
%  enjoys the benefit that secret keys can be obtained no matter how strong the eavesdropping channels are.}

Existing works have demonstrated that  user cooperation can effectively enlarge the  key capacity
 by introducing  additional helper nodes for cooperative key generation \cite{csiszar2000common
 ,lai2012cooperative,zhou2014secret}. The work in \cite{csiszar2000common}
  first studied cooperative key generation (including the generation of secret keys and private keys) in a  single-helper discrete
   memoryless source (DMS) model, where
  the private key  needs to be protected not only from the eavesdropper but also from all the helper node.
  The works in \cite{lai2012cooperative,zhou2014secret} utilized estimates of wireless channels for the key generation
   in cooperative wireless networks, in which
   the relay nodes provide additional resources of wireless fading channels. In \cite{lai2012cooperative},
     a relay-assisted algorithm was proposed to enhance the secret key rate for the scenario  without secrecy constraints at relays,
     and then  an XOR-based algorithm was proposed to generate a relay-oblivious key, (i.e., private key). In \cite{zhou2014secret}, a multi-antenna relay was  considered to
  help the legitimate terminals to generate a secret key, and then
  the optimal attacker's strategy was characterized to minimize the secret key rate when Eve
  is an active attacker.

   The problem of private key generation is investigated in this paper.
    We  consider a particular cooperative PIN model with $M+2$ terminals (Alice, Bob,  $M$ relays)
    and an eavesdropper (Eve), where $M\geq 2$.
  Under the help of relays,  Alice and Bob wish to establish a private key which should be protected from  not only Eve but also from individual relays simultaneously.
   One of the main contributions of this paper is to find the private key capacity of this PIN model. To obtain the lower bound,
    we propose a  novel algorithm for generating the private key. Specifically, using the observations at relays and the transmissions over the public channel, Alice and Bob first agree on $M$ common messages, each of which is open to a certain relay.
Then  a random binning process  is adopted in the  key distillation step to  map these insecure common messages into a private key.
Such an algorithm is termed as the ``RB-based algorithm" for simplicity.
On the other hand, the upper bound of the private key capacity is obtained by considering $M$ enhanced source models, each of which
relaxes the secrecy constraints on some relays, and assumes that the relay observations are known by  Alice or Bob
in advance. Such an upper bound is tight and matches with the lower bound.

The proposed RB-based private key generation algorithm
 in the PIN model can be extended to more practical wireless communications. In particular,  we  consider a
cooperative wireless network, in which Alice, Bob and the $M$ relays use
estimates of wireless channels  as the correlative source observations. It is assumed that Alice and Bob are far
 away from each other, so {\em there does not exist the direct
link} between Alice and Bob.  Compared to the  XOR-based algorithm
 in \cite{lai2012cooperative} whose multiplexing gain
is $\lfloor M \rfloor/2$ for the considered wireless network, the proposed RB-based algorithm achieves a larger multiplexing gain  $M-1$.
%This is because the proposed approach pays less price for
%    satisfying  the secrecy constraints at the relays by using the  random binning process to
%    simultaneously confuse all the relays.

  \vspace{-1em}
\section{Pairwise Independent Network Model}\label{ii}

Consider a  DMS model, where  Alice and Bob, with the help of $M\geq 2$  relays, wish to establish a
private  key
that needs to be  protected  from  Eve and individual relays simultaneously. All relays are assumed to be curious but honest: they will
comply with the proposed transmission schemes for helping Alice and Bob to generate a key, but would also try to intercept
the key information if they can \cite{lai2012cooperative}. The nodes  can communicate to each other over a noiseless public channel
whose capacity is infinite, but the transmitted information over the public channel is also available to Eve noiselessly.
%Note that these assumptions are made only for the DMS model, and will be relaxed in the next section related to wireless networks.
Eve is passive in the sense that it only receives but not transmits information.

\vspace{-0.3em}
 \begin{figure}[htbp]\centering
    \epsfig{file=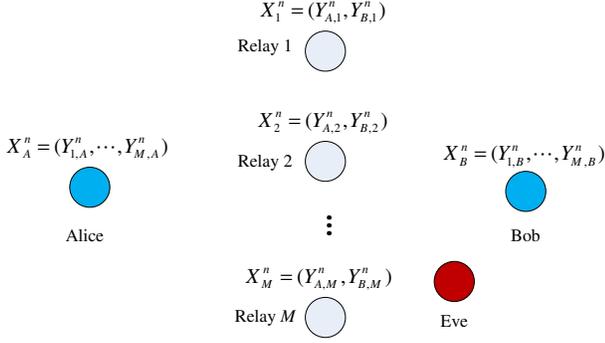, width=0.45\textwidth,clip=}
\caption{The considered cooperative PIN model with $M$ relays.}\label{pair-wise}
\end{figure}
\vspace{-0.5em}

{ For $\forall m\in\{1,\cdots,M\}$, let $Y_{m,A}$  and $Y_{A,m}$ denote the correlative source observations at  Alice and relay $m$, respectively.
$Y_{m,B}$  and $Y_{B,m}$ denote the correlative source observations at Bob and relay $m$, respectively.} Specifically,
 Alice observes $n$ i.i.d. repetitions of random variable
$X_A=(Y_{1,A},\cdots,Y_{M,A})$, denoted by $X_A^n=(Y_{1,A}^n,\cdots, Y_{M,A}^n)$; Bob observes $n$ i.i.d. repetitions of random variable
$X_B=(Y_{1,B},\cdots,Y_{M,B})$, denoted by $X_B^n=(Y_{1,B}^n,\cdots, Y_{M,B}^n)$; relay $m$
observes
$n$ i.i.d. repetitions of random variable
$X_m=(Y_{A,m},Y_{B,m})$, denoted by $X_m^n=(Y_{A,m}^n,Y_{B,m}^n)$. This DMS model  is a PIN
 in the sense that \cite{ye2007group}
\begin{align}
  I(Y_{i,\alpha},Y_{\alpha,i};
  \{Y_{j,\beta},Y_{\beta,j}:(j,\beta)\neq (i,\alpha)\}
  )=0,\nonumber\\
  \textrm{ for } i,j\in\{1,\cdots,M\};
  \alpha, \beta\in\{A,B\}.
\end{align}
This means that Alice and relay $m$ have access  to a pair $(Y_{m,A},Y_{A,m})$ which is independent of any other
pair of source observations, so is $(Y_{m,B},Y_{B,m})$. Note that there does not exist correlated source
 observations between Alice and Bob, the private key can be generated only via the help from the relays.
 Moreover, we do not consider correlated sources observed by any pair of relays, since the common randomness
 shared by any pair of  relays cannot contribute to the private key rate.

More definitions  are given as follows.
\begin{itemize}
  \item[$\bullet$] Without loss of generality, assume that the nodes use the public channel to communicate
  in a round robin fashion over $q$ rounds. Let $1\leq l \leq q$ and $1\leq m \leq M$.
  %, where $q$ is the integral multiple of $M+2$, i.e. $q=q_0(M+2)$.
  Specifically,
  relay $m$  transmits during rounds $l$ that satisfy  $ l\mod (M+2)=m$; Alice transmits
  during rounds $l$ that satisfy $  l\mod (M+2)=M+1$; Bob transmits
  during rounds $l$ that satisfy $  l\mod (M+2)=0$.
  \item[$\bullet$] A $(2^{n\tilde{R}_1},\cdots,2^{n\tilde{R}_q})$ code for the cooperative key generation problem
  consists of :
%$l=m,(M+2)+m,\cdots, (q_0-1)(M+2)+m$$l=M+1, 2M+1,\cdots, q_0(M+2)-1$$l=M+2,2(M+2),\cdots, q_0(M+2)$

  $\quad$(i) $M+2$ randomized encoders, one for each node. In  rounds $l$ satisfying $ l\mod (M+2)=m$,
  relay $m$ generates an index $F_l\in \{1,\cdots,2^{{n\tilde{R}_l}}\}$ according to $p(f_l|x_m^n,f^{l-1})$; in  rounds $l$ satisfying $ l\mod (M+2)=M+1$,
  Alice generates an index $F_l\in \{1,\cdots,2^{{n\tilde{R}_l}}\}$ according to $p(f_l|x_A^n,f^{l-1})$; in  rounds $l$ satisfying $ l\mod (M+2)=0$,
  Bob generates an index $F_l\in \{1,\cdots, 2^{{n\tilde{R}_l}}\}$ according to $p(f_l|x_B^n,f^{l-1})$.

  $\quad$(ii) Two decoders, one for Alice (decoder 1) and the other for Bob (decoder 2). {After receiving
  the $q$ rounds of transmissions (i.e., $F^q=\{F_1,\cdots,F_q\}$) over the public channel,}
  decoder 1 generates a
  random key $K_A$ according to $K_A=K_A(X_A^n,F^q)$; Decoder 2 generates a
  random key $K_B$ according to $K_B=K_B(X_B^n,F^q)$.
   %\item[$\bullet$] The probability of error for the code is $P_e^{(n)}=Pr(K_A\neq K_B)$.
    \item[$\bullet$] A private key rate $R$ is said to be {\em achievable} if there exists a $(2^{n\tilde{R}_1},\cdots,2^{n\tilde{R}_q})$ code such that
        \begin{align}
          &Pr(K_A\neq K_B)\leq \epsilon,\label{decoder-requirement}\\
          &\frac{1}{n}H(K_A)\geq R-\epsilon,  \label{key-rate-requirement}\\
          & \frac{1}{n}H(K_A)\geq \frac{1}{n} \log |\mathcal{K}_A|-\epsilon, \label{uniformly}\\
          &\frac{1}{n}I(K_A;X_m^n,F^q)\leq \epsilon,\textrm{ for } \forall m\in\{1,\cdots,M\},\label{key-requirment}
        \end{align}
      where $|\mathcal{K}_A|$ denotes the size of the alphabet of the key $K_A$. Note that
       the secrecy constraints in \eqref{key-requirment} implies that the relays are  assumed to be {\em non-colluding}.
%       ensure that $K_A$ is secret from
%        the Eve, i.e., $\frac{1}{n}I(K_A;F^q)\leq \epsilon$ can be obtained from \eqref{key-requirment}.
   \item[$\bullet$] The private key capacity $C_K$ is the supremum of all achievable rates $R$.
   $C_K^{(d)}$ is used to denote the private
    key capacity with deterministic encoding and key generation functions.
   According to \cite{csiszar2000common}, $C_K^{(d)}=C_K$, which means that randomization is useless for key
   generation in the addressed source model.
\end{itemize}

\section{Private Key Capacity of the PIN model}\label{iii}
For simplicity, we first define
\begin{align}\label{I_i}
  I_i=\min\left\{I(Y_{i,A};Y_{A,i}),I(Y_{i,B};Y_{B,i})\right\}, \forall i\in \{1,\cdots, M\}.
\end{align}
Furthermore, these parameters are ordered according to $I_{(1)}\leq I_{(2)} \leq \cdots \leq I_{(M)}$. Then
the private key capacity for the considered scenario is given in the following theorem.
\begin{theorem}\label{theorem_key_capacity}
  For the considered  PIN model with $M$ relays, the private key capacity is given by
  \begin{align}
    C_K=&\sum_{i=1}^M I_{i}- \max_{m\in\{1,\cdots, M\}}I_{m} \label{C_K(1)}\\
    =&\sum_{i=1}^{M-1} I_{(i)}.
  \end{align}
\end{theorem}
\begin{proof}
The achievability part is proved by a novel RB-based key generation algorithm that is based on two steps: key agreement
and  key distillation.
In the key agreement step, Alice and Bob can agree on $M$ common messages, each of which is revealed to a certain relay.
In the private key distillation step, these common messages will be mapped into the final private key via a RB-based private-key codebook.

The converse part is proved by deriving the  upper bounds of $M$ symmetric enhanced channels. Each of these enhanced channels
relaxes the secrecy constraints and assumes Alice and Bob to be genie-aided (i.e., knowing part of the sources
observed by the relays).

The details of the proof will be provided as follows.
\end{proof}

%   Obviously, the upper bound of the above enhanced source model is also the upper bound of the addressed source
%   model in Section \ref{ii}. Considering all the $M$ enhanced source models, we can choose the minimum value
%   among them as the final tight upper bound on the secret key capacity of the addressed model.
 %It can be shown that the obtained upper bound matches with the achievable key rate, which completes the proof for the private key capacity theorem.

\vspace{-1em}
\subsection{Proof of Achievability}\label{achievablility}
Algorithm 1 briefly shows the  achievable scheme that is based on two steps: key agreement and  key distillation.
Let $R_{A,i}=I(Y_{A,i},Y_{i,A})-\epsilon$, $R_{B,i}=I(Y_{B,i},Y_{i,B})-\epsilon$  for $1\leq i \leq M$; $R_i=\min\{R_{i,A},R_{i,B}\}=I_i-\epsilon$, and they are ordered according to
$R_{(1)}\leq \cdots \leq R_{(M)}$. Besides, $R_{key}=\sum_{i=1}^{M-1} R_{(i)}$.
  \begin{figure}[htbp]

\hrulefill

{\bf Algorithm 1:}  Algorithm of Relay-oblivious Key Generation

    \hrulefill

 Step 1: Key Agreement:
 \begin{enumerate}[$\quad\bullet$]
   \item Alice and Relay $i$  agree on a pairwise key $W_{A,i}$ from the correlated observations $(Y_{i,A}^n,Y_{A,i}^n)$;
    Bob and Relay $i$  agree on a pairwise key $W_{B,i}$ from the correlated observations $(Y_{i,B}^n,Y_{B,i}^n)$,
   $i=1,\cdots M$.
   \item  Relay $i$ sends $W_{A,i} \oplus W_{B,i}$  over the public channel, so Alice and Bob can obtain both  $W_{A,i}$
   and $ W_{B,i}$,   $i=1,\cdots M$. Then they will choose the one with a smaller size as the common message, denoted as $W_i\in\mathcal{W}_i$.
  \end{enumerate}

 Step 2: Key Distillation:
\begin{enumerate} [   $\ \; \bullet$]
\item In advance,  randomly  grouped all the sequences $w^M$ in $\mathcal{W}^M$ into $2^{n(R_{key}-\epsilon)}$ bins each with
equal amount of  codewords. All the other nodes also know this private-key codebook.
   \item Alice and Bob find the sequence  $W^M=(W_1,\cdots,W_M)$ in the RB based private-key codebook,
     and choose its bin number  as the final private key.
     \end{enumerate}

 \hrulefill
    \end{figure}

\subsubsection{Key Agreement}
In the key agreement step,  Alice and Bob will agree on $M$ common messages.

 First, each relay $i$ and Alice agree on a pairwise key $W_{A,i}$  using their correlated sources $(Y_{A,i}^n,Y_{i,A}^n)$;
 each relay and Bob agree on a pairwise key  using  their correlated sources $(Y_{B,i}^n,Y_{i,B}^n)$.
According to the standard techniques \cite{ahlswede1993common}\cite{nitinawarat2010secret},
 each pairwise key  $W_{A,i}$  ($W_{B,i}$) is generated using  Slepian-Wolf coding and public
transmission $F_{A,i}$  ($F_{B,i}$). Moreover, the pairwise keys  $W_{A,i}$ and $W_{B,i}$ have the following properties\cite{ye2007group,nitinawarat2010secret}:
\begin{itemize}
\item[i)] They
can achieve the rates
$R_{A,i}$ and $R_{B,i}$, respectively;
\item[ii)] They are uniformly distributed and can be decoded by
both Alice and Bob correctly;
\item[iii)] The pairs $\{(W_{\alpha,i},F_{\alpha,i})_{\alpha\in\{A,B\},i\in\{1,\cdots,M\}}\}$
are mutually independent, due to the definitions of  the PIN model.
\end{itemize}

    Second, each relay $i$ sends out  $W_{A,i} \oplus W_{B,i}$
     over the public channel, so Alice and Bob can obtain both the two pairwise keys, and choose the one with a smaller size as the common message, denoted as $W_i$.
Hence the rate of each common message $W_i$ is $R_i$. According to \cite{lai2012cooperative},
\begin{align}\label{security_constraint}
  \frac{1}{n} I(W_1,\cdots,W_M;F^q)\leq \epsilon_1.
\end{align}
\subsubsection{Key Distillation}
 In the  key distillation step,  both Alice and Bob map  all the insecure
   common messages  assembled  from the key agreement step into the unique codeword in the private-key
   codebook, and set the bin number of this codeword as
  the final private key. Note that such a private-key codebook is generated based on random binning, so it provides necessary randomness such that  the bin number is secret from all the relays and Eve.
  \begin{Remark}
 The main difference between the proposed algorithm and the one in \cite{lai2012cooperative} lies in the key distillation step:
 the former is based on the RB process and the latter is based on an XOR process.
 In \cite{lai2012cooperative},  Alice and Bob concatenate $(W_1\oplus W_2,\cdots,
   W_{M-1}\oplus W_{M})$ as the final private key in the key distillation step. Here $M$ is assumed
    to be even.
    \end{Remark}

We will provide more details of the RB-based codebook in the following.

{\bf Codebook Generation}:
Let $w_i\in \mathcal{W}_i=\{1,\cdots,2^{nR_i}\}$,  $w^M=(w_1,\cdots,w_M)$.
 %where $w_i=w_{A,i}$ if $R_i=R_{A,i}$, and $w_i=w_{B,i}$ if $R_i=R_{B,i}$.
 Then, based on the concept of random binning,  the private-key codebook can be constructed.
 Specifically,   randomly and uniformly partition all
the elements $w^M$ in set
$\mathcal{W}^M=\mathcal{W}_1\times\mathcal{W}_2
  \times\cdots \times \mathcal{W}_M$ into $2^{n(R_{key}-\epsilon)}$ bins each with $2^{n(R_{(M)}+\epsilon)}$ codewords.
  So each codeword $w^M$ can be indexed as $w^M(k,\tilde{k})$, where
  $k\in\{1,\cdots,2^{n(R_{key}-\epsilon)}\}$, $\tilde{k}\in\{1,\cdots,2^{n(R_{(M)}+\epsilon)}\}$. Fig. \ref{random-binning}
  illustrates the binning assignment for the private-key codebook, denoted by $\mathcal{C}$, that is known by
all nodes (including Eve).

\vspace{-0.3em}
\begin{figure}[htbp]\centering
    \epsfig{file=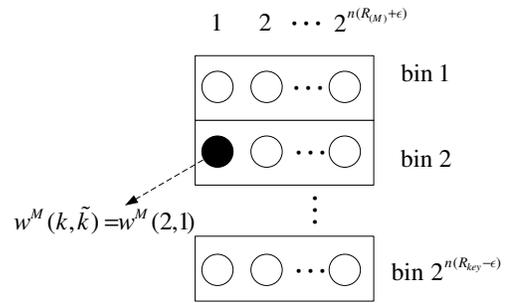, width=0.37\textwidth,clip=}
\caption{The binning assignment for the  private-key codebook, where $w^M=(w_1,\cdots, w_M)\in\mathcal{W}^M$, $w_i\in \{1,\cdots,2^{nR_i}\}$.}\label{random-binning}
\end{figure}
\vspace{-0.5em}

{\bf Decoding and key generation:} Based on the common messages collected in the key agreement step, Alice and Bob
 can find their corresponding indices in the private-key codebook.
  Specifically, knowing $(W_1,\cdots,W_M)$, Alice finds the index pair $(k,\tilde{k})$ from the
   private-key codebook such that $w^M(k,\tilde{k})=(W_1,\cdots,W_M)$. Then, it sets its key $K_A=k$.
Similarly, Bob can also correctly find
the key $K_B=k$. Since the error probability of the event that Alice and Bob share the same $(W_1,\cdots,W_M)$
 is insignificant, the error probability $P(K_A
\neq K_B)$ is arbitrarily small as $n\rightarrow \infty$.

 {\bf Analysis of the key rate}:
  %Due to the Slepain-Wolf coding employed in the achievable scheme,
  %$W_i$ is nearly uniformly distributed over $\mathcal{W}_i$, $i=1,\cdots, M$,
% so $W^M$ is also nearly uniformly distributed over $\mathcal{W}^M$.
  Since  the private-key codebook
 is  based on the random binning process,
$K_A$ is  uniformly distributed over $\{1,\cdots, 2^{n(R_{key}-\epsilon)}\}$ averaged over the codebook
(i.e., $\mathcal{C}$).
 Therefore, it can be obviously obtained that $H(K_A|\mathcal{C})=n(R_{key}-\epsilon)$.

  {\bf Analysis of the secrecy constraints}: For  $\forall m\in\{1,\cdots,M\}$, we will prove that the generated private key is secret from relay
 $m$.
 Define $W^M=(W_1,\cdots,W_M)$.
 Then,  averaged over $\mathcal{C}$,  we have
 \begin{align}\label{equivocation1}
I(K_A;F^q, X_m^n|\mathcal{C})& \stackrel{(a)}{\leq} I(K_A;F^q, W_m|\mathcal{C})  \nonumber\\
    &\leq  I(K_A; W_m|\mathcal{C}) + I(K_A,W^M;F^q|W_m,\mathcal{C}) \nonumber\\
 &\stackrel{(b)}{\leq}  I(K_A; W_m|\mathcal{C}) + n\epsilon_1 \nonumber\\
  &=I(K_A;W_m|\mathcal{C})+n\epsilon_1.
  \end{align}
  where $(a)$ is due to the fact that $X_m^n-(W_m,F^q)-K_A$ is a Markov chain;
  $(b)$ is due to \eqref{security_constraint} and the fact that $K_A$ is determined by $W^M$ for a given codebook.
  Furthermore,
  \begin{align}
    &I(K_A;W_m|\mathcal{C})\!=\!I(K_A,W^M;W_m|\mathcal{C})\!-\!I(W^M;W_m|K_A,\mathcal{C})\nonumber\\
    &\!=\!I(W^M;W_m|\mathcal{C})\!-\!H(W^M|K_A,\mathcal{C})\!+\!H(W^M|W_m,K_A,\mathcal{C})\nonumber\\
     &\!=\!H(W_m|\mathcal{C})\!-\!H(W^M|K_A,\!\mathcal{C})\!+\!H(W^M|W_m,\!K_A,\!\mathcal{C}).\label{equivocation}
  \end{align}
 For the first term, obviously we have
  \begin{align}\label{term_one}
    H(W_m|\mathcal{C})= nR_m.
  \end{align}
 Since   $H(W_i|\mathcal{C})= nR_i$, we have $H(W^M|\mathcal{C})= n\sum_{i=1}^MR_{(i)}$. So the second term can be obtained as
  \begin{align}\label{term_two}
    H(W^M|K_A,\mathcal{C})&=H(W^M|\mathcal{C})+H(K_A|W^M,\mathcal{C})-H(K_A|\mathcal{C})\nonumber\\
    &=H(W^M|\mathcal{C})-H(K_A|\mathcal{C})\nonumber\\
    %&\geq H(W^M)-H(K_A)\nonumber\\
    &= n\sum_{i=1}^MR_{(i)}-n(R_{key}-\epsilon)\nonumber\\
    &=n(R_{(M)}+\epsilon).
  \end{align}
  The third term is bounded in the following lemma.
  \begin{lemma}\label{lemma_thirdterm}
    When $R_{(M)}=\max\{R_1,\cdots,R_M\}$ and $n$ is sufficiently large,
    \begin{align}\label{term_3}
    H(W^M|W_m,K_A,\mathcal{C})\leq n(R_{(M)}-R_m+\delta(\epsilon)). \end{align}
  \end{lemma}
  \begin{proof}
 This lemma can be proved using similar methods in  existing related works, such as \cite{el2011network} (proof of Lemma 22.3)  and \cite{zhang2014capacity}, with some necessary variations. The details are omitted here due to
 space limitation.
  \end{proof}

  Combining \eqref{equivocation1} with \eqref{equivocation}, \eqref{term_one}, \eqref{term_two} and \eqref{term_3}, we have
  \begin{align}\frac{1}{n}I(K_A;F^q, X_m^n|\mathcal{C})&\leq\frac{1}{n}I(K_A;W_m|\mathcal{C})+\epsilon_1\nonumber\\
 &\leq \delta(\epsilon)+\epsilon_1-\epsilon. \end{align}
 So the private key rate $R_{key}=\sum_{i}^{M-1}I_{(i)}-\epsilon$ is achievable.

\vspace{-1em}
  \subsection{Proof of Converse}
The calculation of the upper bound is based on $M$ symmetric enhanced channels. For the $m$-th enhanced source model, $m=1,\cdots, M$,
    we only consider the secrecy constraint on relay $m$,
    and ignore the secrecy constraints on all the other relays.
  %  Alice and Bob are assumed to know the observations of relay $m$ (i.e., $X_m^n$) a priori.
    Moreover, Alice and Bob are assumed to know  the  observations of two
     subsets of relays a priori, respectively. The definitions of the two subsets are given as follows.

For a given  $m\in\{1,\cdots, M\}$, we will form two sets of nodes, i.e., $\mathcal{A}_{]m[}$ and $\mathcal{B}_{]m[}$ in the next.
First, allocate Alice and
Bob to
$\mathcal{A}_{]m[}$ and $\mathcal{B}_{]m[}$, respectively.
Second, for relay $i$, $i\neq m$, if $I(Y_{A,i};Y_{i,A})>I(Y_{B,i};Y_{i,B})$, allocate it
to $\mathcal{A}_{]m[}$; otherwise, allocate it to $\mathcal{B}_{]m[}$. So $I(Y_{B,i};Y_{i,B})=\min\{I(Y_{A,i};Y_{i,A}),I(Y_{B,i};Y_{i,B})\}$ if relay $i$ lies in $\mathcal{A}_{]m[}$, and  $I(Y_{A,i};Y_{i,A})=\min\{I(Y_{A,i};Y_{i,A}),I(Y_{B,i};Y_{i,B})\}$ if relay $i$ lies in $\mathcal{B}_{]m[}$.

Then, assume without loss of generality that relays $A_1$, $A_2$, $\cdots$, $A_j$ are allocated to $\mathcal{A}_{]m[}$,
and relays $B_{1}$, $B_{2}$, $\cdots$, $B_{M-1-j}$ are allocated to $\mathcal{B}_{]m[}$,
$0\leq j \leq M-1$\footnote{If $j=0$, $\{A_1,\cdots,A_j\}=\emptyset$ and  $\mathcal{A}_{]m[}$ = \{Alice\};
if $j=M-1$, $\{B_1,\cdots,B_{M-1-j}\}=\emptyset$ and  $\mathcal{B}_{]m[}$ = \{Bob\}.}. Here
  $\{A_1,\cdots,A_j\}\bigcap \{B_{1},\cdots,B_{M-1-j}\}=\emptyset$ and
   $\{A_1,\cdots,A_j\}\bigcup \{B_{1},\cdots,B_{M-1-j}\}=\{1,\cdots,m-1,m+1,\cdots, M\}$.
In other words, $\mathcal{A}_{]m[}$ = \{Alice, relays $A_1$,  $\cdots$, $A_j$\};
$\mathcal{B}_{]m[}$ =\{Bob, relays $B_{1}$,  $\cdots$, $B_{M-1-j}$\}.
Now, by the max-flow principle  \cite{ye2007group}, the max follow between the two sets $\mathcal{A}_{]m[}$
and $\mathcal{B}_{]m[}$ can be expressed as $\sum_{i=1}^m I_i -I_m$, which is
the upper bound of the $m$-th enhanced channel. Due to space limitation, the details are omitted here.

Choosing the smallest bounds among all the $M$ enhanced channels,
we can obtain $C_K\leq \sum_{i=1}^M I_i -\max_{m\in\{1,\cdots,M\}}I_m$.
%
%\begin{figure}[htbp]\centering
%    \epsfig{file=two_sets.eps, width=0.5\textwidth,clip=}
%\caption{The $m$-th enhanced source model, where
%Alice and Bob know the observations at the nodes in  sets $\mathcal{A}_{]m[}$
%and  $\mathcal{B}_{]m[}$ respectively  when considering the
%secret constraint at relay $m$. }\label{two-sets}
%\end{figure}
\vspace{-1em}
\section{Key Generation in Wireless Network}\label{iv}
%\begin{figure}[htbp]\centering
%    \epsfig{file=wireless_network.eps, width=0.45\textwidth,clip=}
%\caption{A system diagram for the cooperative wireless network with secrecy constraints at relays}\label{wireless}
%\end{figure}

In this section, we will extend the   RB-based algorithm proposed for the PIN model
into the wireless network, and use the estimates of wireless fading channels as
source observations for private key generation.
\vspace{-1em}
\subsection{Model}
The considered wireless network can be
 viewed as a practical example of the PIN model in
Section \ref{ii}. All the nodes have a single antenna and are half-duplex constrained.
 In this wireless network, it is assumed that there is no direct link     between
Alice and Bob, since they are located far from each other.
Denote $h_{A,i}$  ($h_{B,i}$)   as the  fading channel gains between  relay $i$  and Alice (Bob).
All channels are assumed to be reciprocal. It is reasonable to assume that all the fading channel gains and noise are  random variables and
 independent of each other. An ergodic {\em block fading} model is considered, in which all the channel gains remain
  constant for a block of $T$ symbols and change randomly to other independent values after the current
   block. For simplicity, we assume $h_{A,i}\sim \mathcal{N}(0,\delta_{A,i})$, $h_{B,i}\sim \mathcal{N}(0,\delta_{B,i})$.
   %, where our scheme can be extended to   other fading channel models in a straightforward manner.
   Moreover, none of the nodes knows the values of $h_{A,i}$ and $h_{B,i}$ a priori, but all the nodes know
   their statistics.

Assume that terminals transmit in a time-division manner.
  For  $L$ channel uses, let $\mathbf{S}_i=[s_i(1),\cdots,s_i(L_i)]^T$,  $\mathbf{S}_A=[s_A(1),\cdots,s_A(L_A)]^T$ and
   $\mathbf{S}_B=[s_B(1),\cdots,s_B(L_B)]^T$ denote the signals sent by relay $i$,
   Alice and Bob, respectively, where $i=1,\cdots, M$, and $L_A+L_B+\sum_{i=1}^ML_i=L$. For simplicity, we consider an equal power constraint
   for the legitimate terminals, that is
   \begin{align}
     \frac{1}{L_i} \mathbb{E}\{\mathbf{S}_i^T\mathbf{S}_i\},
     \frac{1}{L_A} \mathbb{E}\{\mathbf{S}_A^T\mathbf{S}_A\},
     \frac{1}{L_B} \mathbb{E}\{\mathbf{S}_B^T\mathbf{S}_B\}\leq P.
   \end{align}

\vspace{-1em}
\subsection{Proposed RB-based Algorithm}
The proposed RB-based  algorithm (Algorithm 1) can be extended to wireless networks for private key generation. Briefly speaking,
all the relays, Alice and Bob take turns to broadcast training sequences. After the channel estimation step,
Alice and Bob will generate the private key using the RB-based scheme in Algorithm 1 (Section \ref{iii}).
  Now, we will explain the channel estimation step in more detail. Fig. \ref{training} shows the time frame
  for the training of the proposed scheme in each fading block. Each fading block is divided into $M+2$ time slots, and
  the numbers of symbols in these  time slots are  $T_1,\cdots,T_M$, $T_A$, $T_B$,  respectively,
  where $T_A+T_B+\sum_{i=1}^MT_M=T$. Suppose relay $i$, Alice and Bob sends known training sequences $\mathbf{S}_i$ of size
  $1\times T_i$, $\mathbf{S}_A$ of size
  $1\times T_A$ and $\mathbf{S}_B$ of size
  $1\times T_B$, respectively. The energy of each sequence is $||\mathbf{S}_i||^2=T_iP$,  $||\mathbf{S}_A||^2=T_AP$,
  $||\mathbf{S}_B||^2=T_BP$, where $||\cdot||$ denotes the Euclidean norm.
  %\vspace{-0.3em}
\begin{figure}[htbp]\centering
    \epsfig{file=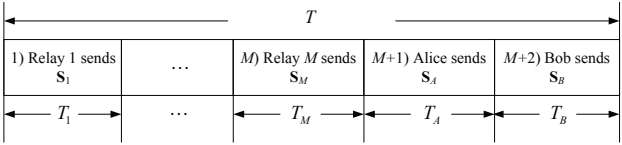, width=0.47\textwidth,clip=}
\caption{Time frame for the training phase  in each fading block.}\label{training}
\end{figure}
 %\vspace{-0.5em}
%  For channel $h_{A,i}$, after the training phase in each block, Alice and relay $i$ receive the following:
%  \begin{align}
%    \mathbf{Y}_A^{(i)}&=h_{A,i}\mathbf{S}_i+\mathbf{N}_{A}^{(i)},\\
%     \mathbf{Y}_i^{(M+1)}&=h_{A,i}\mathbf{S}_A+\mathbf{N}_{i}^{(M+1)},
%  \end{align}
%  respectively. Then, Alice and relay $i$ can obtain the estimates as following:
%  \begin{align}
%    \tilde{h}_{i,A}&=\frac{\mathbf{S}_i^T}{||\mathbf{S}_i||^2}\mathbf{Y}_A^{(i)}
%    =h_{A,i}+\frac{\mathbf{S}_i^T}{||\mathbf{S}_i||^2}\mathbf{N}_A^{(i)},\\
%    \tilde{h}_{A,i}&=\frac{\mathbf{S}_A^T}{||\mathbf{S}_A||^2}\mathbf{Y}_i^{(M+1)}
%    =h_{A,i}+\frac{\mathbf{S}_A^T}{||\mathbf{S}_A||^2}\mathbf{N}_A^{(M+1)}.
%  \end{align}
% Similarly, for channel $h_{B,i}$, Bob and relay $i$ can obtain the estimates as following:
% \begin{align}
%    \tilde{h}_{i,B}&=\frac{\mathbf{S}_i^T}{||\mathbf{S}_i||^2}\mathbf{Y}_B^{(i)}
%    =h_{B,i}+\frac{\mathbf{S}_i^T}{||\mathbf{S}_i||^2}\mathbf{N}_B^{(i)},\\
%    \tilde{h}_{B,i}&=\frac{\mathbf{S}_B^T}{||\mathbf{S}_B||^2}\mathbf{Y}_i^{(M+2)}
%    =h_{B,i}+\frac{\mathbf{S}_B^T}{||\mathbf{S}_B||^2}\mathbf{N}_B^{(M+2)}.
%  \end{align}

From $n$ fading blocks, Alice can obtain the estimates $(\tilde{h}_{1,A}^n,\cdots,\tilde{h}_{M,A}^n)$;
  Bob can obtain the estimates $(\tilde{h}_{1,B}^n,\cdots,\tilde{h}_{M,B}^n)$;
  relay $i$ can obtain the estimates $(\tilde{h}_{A,i}^n,\tilde{h}_{B,i}^n)$, $i=1,\cdots, M$.
  These estimates are noisy versions of the corresponding fading channels. The details of this channel estimation step
   are omitted here due to space limitation, and similar works can be found in \cite{lai2012cooperative,zhou2014secret}.
  The rate of each pairwise key $W_{\alpha,i}$ can be calculated as
         \begin{align}\label{I_i^G}
        R_{\alpha,i}^G&=\frac{1}{2} \log\left(1+\frac{T_iT_{\alpha}P^2\delta_{\alpha,i}^4}{\delta^4+(T_i+T_{\alpha})
        \delta^2\delta_{\alpha,i}^2P}\right), \nonumber\\
        &\quad \forall \alpha\in\{A,B\}, i\in\{1,\cdots, M\},
      \end{align}
      where $\delta^2$ is the variance of each Gaussian noise.

Now, using the result in
      Theorem \ref{theorem_key_capacity}, the proposed RB-based algorithm achieves the private key rate $R_{key}^G$  for some tuple $(T_A,T_B,T_1,\cdots,T_M)$, which
       can be written as
      \begin{align}\label{R_kG}
        R_{key}^G=\frac{1}{T}\left(\sum_{i=1}^M I_i^G -\max_{i\in\{1,\cdots,M\}} I_i^G\right),
      \end{align}
      with  $I_i^G=\min\{R_{A,i}^G,R_{B,i}^G\}$.
      To further show the impact of the proposed scheme on the gain of the key rate, the multiplexing gain (introduced in \cite{lai2012cooperative}) is analyzed as following.
  \begin{corollary}
         For the considered wireless network with  $M$ relays, the  multiplexing gain of the private key rate
         achieved by the proposed         RB-based algorithm is   $M-1$.
      \end{corollary}
      \begin{proof}
        Based on the definition of in \cite{lai2012cooperative}, the multiplexing gain of the proposed algorithm should be
        $\lim_{P\rightarrow \infty}\frac{R_{key}^G}{R_s}$, where $R_s\approx \frac{1}{2T}\log P$ as ${P\rightarrow \infty}$.
        From Eq. \eqref{I_i^G},  it is easy to obtain that  $\lim_{P\rightarrow \infty}\frac{ R_{\alpha,i}^G}{R_s}=T$, so we have
        $\lim_{P\rightarrow \infty}\frac{R_{key}^G}{R_s}=M-1.$
        \end{proof}

      \begin{Remark}\label{Remark_mg}
       If there is no secrecy constraints at the relays, the multiplexing gain is $M$ \cite{lai2012cooperative}.
       So the proposed RB-based algorithm sacrifices one multiplexing gain for
    satisfying the secrecy constraints at all the $M$ relays.  This loss is insignificant  because  {\em only one} multiplexing gain
    is sacrificed, no matter how large $M$ is.
    But for the XOR-based algorithm  in \cite{lai2012cooperative} (Corollary 10), its
    multiplexing gain is $\lfloor M/2 \rfloor$ if there does exist the direct link between Alice and Bob.
    Therefore this existing scheme suffers a loss of $M/2$
    multiplexing gain in comparison with  the case without secrecy constraints at the relays. Hence the
    proposed RB-based scheme can effectively enhance the performance of the private key generation.
      \end{Remark}
\vspace{-0.5em}
\section{Conclusion}\label{vi}
 In this paper, we have investigated the problem of private key generation. A particular cooperative PIN model with $M+2$
 terminals is considered, where Alice, Bob and $M$ relays observe pairwise independent sources.
  Under the help of relays,  Alice and Bob wish to establish a private key that is secure from Eve and all relays.
  The private key capacity of this PIN model has been found. The achievability
  is proved via  a novel RB-based algorithm for generating the private key.
The upper bound of the private key capacity is obtained by considering $M$ enhanced source models.
Then,   we further consider a
cooperative wireless network, in which
estimates of wireless channels are regarded as the correlative source observations.
Compared to the  XOR-based algorithm
 in \cite{lai2012cooperative} whose multiplexing gain
is $\lfloor M \rfloor/2$, the proposed RB-based algorithm achieves a larger multiplexing gain  $M-1$.

\section*{Acknowledgements}
 The work of Peng Xu and Xuchu Dai was supported by the National Natural Science Foundation of China (NSFC)
under grant number 61471334.

\vspace{-0.5em}
\bibliographystyle{IEEEtran}
\bibliography{references}

\begin{thebibliography}{10}
\providecommand{\url}[1]{#1}
\csname url@samestyle\endcsname
\providecommand{\newblock}{\relax}
\providecommand{\bibinfo}[2]{#2}
\providecommand{\BIBentrySTDinterwordspacing}{\spaceskip=0pt\relax}
\providecommand{\BIBentryALTinterwordstretchfactor}{4}
\providecommand{\BIBentryALTinterwordspacing}{\spaceskip=\fontdimen2\font plus
\BIBentryALTinterwordstretchfactor\fontdimen3\font minus
  \fontdimen4\font\relax}
\providecommand{\BIBforeignlanguage}[2]{{%
\expandafter\ifx\csname l@#1\endcsname\relax
\typeout{** WARNING: IEEEtran.bst: No hyphenation pattern has been}%
\typeout{** loaded for the language `#1'. Using the pattern for}%
\typeout{** the default language instead.}%
\else
\language=\csname l@#1\endcsname
\fi
#2}}
\providecommand{\BIBdecl}{\relax}
\BIBdecl

\bibitem{ye2007group}
C.~Ye and A.~Reznik, ``Group secret key generation algorithms,'' in \emph{IEEE
  International Symposium on Information Theory}, 2007, pp. 2596--2600.

\bibitem{nitinawarat2010secret}
S.~Nitinawarat, C.~Ye, A.~Barg, P.~Narayan, and A.~Reznik, ``Secret key
  generation for a pairwise independent network model,'' \emph{IEEE
  Transactions on Information Theory}, vol.~56, no.~12, pp. 6482--6489, 2010.

\bibitem{lai2012simultaneously}
L.~Lai and S.-W. Ho, ``Simultaneously generating multiple keys and
  multi-commodity flow in networks,'' in \emph{IEEE Information Theory Workshop
  (ITW)}, 2012, pp. 627--631.

\bibitem{lai2013simultaneously}
L.~Lai and L.~Huie, ``Simultaneously generating multiple keys in many to one
  networks,'' in \emph{IEEE International Symposium on Information Theory
  Proceedings (ISIT)}, 2013, pp. 2394--2398.

\bibitem{csiszar2000common}
I.~Csisz{\'a}r and P.~Narayan, ``Common randomness and secret key generation
  with a helper,'' \emph{IEEE Transactions on Information Theory}, vol.~46,
  no.~2, pp. 344--366, 2000.

\bibitem{csiszar2004secrecy}
------, ``Secrecy capacities for multiple terminals,'' \emph{IEEE Transactions
  on Information Theory}, vol.~50, no.~12, pp. 3047--3061, 2004.

\bibitem{ahlswede1993common}
R.~Ahlswede and I.~Csiszar, ``Common randomness in information theory and
  cryptography. part i: secret sharing,'' \emph{IEEE Transactions on
  Information Theory}, vol.~39, no.~4, 1993.

\bibitem{tse2005fundamentals}
D.~Tse and P.~Viswanath, \emph{Fundamentals of wireless communication}.\hskip
  1em plus 0.5em minus 0.4em\relax Cambridge university press, 2005.

\bibitem{ye2010information}
C.~Ye, S.~Mathur, A.~Reznik, Y.~Shah, W.~Trappe, and N.~B. Mandayam,
  ``Information-theoretically secret key generation for fading wireless
  channels,'' \emph{IEEE Transactions on Information Forensics and Security},
  vol.~5, no.~2, pp. 240--254, 2010.

\bibitem{wang2012cooperative}
Q.~Wang, K.~Xu, and K.~Ren, ``Cooperative secret key generation from phase
  estimation in narrowband fading channels,'' \emph{IEEE Journal on Selected
  Areas in Communications}, vol.~30, no.~9, pp. 1666--1674, 2012.

\bibitem{lai2012cooperative}
L.~Lai, Y.~Liang, and W.~Du, ``Cooperative key generation in wireless
  networks,'' \emph{IEEE Journal on Selected Areas in Communications}, vol.~30,
  no.~8, pp. 1578--1588, 2012.

\bibitem{zhou2014secret}
H.~Zhou, L.~Huie, and L.~Lai, ``Secret key generation in the two-way relay
  channel with active attackers,'' \emph{IEEE Transactions on Information
  Forensics and Security}, vol.~9, no.~3, pp. 476--488, 2014.

\bibitem{el2011network}
A.~El~Gamal and Y.-H. Kim, \emph{Network information theory}.\hskip 1em plus
  0.5em minus 0.4em\relax Cambridge University Press, 2011.

\bibitem{zhang2014capacity}
H.~Zhang, L.~Lai, Y.~Liang, and H.~Wang, ``The capacity region of the
  source-type model for secret key and private key generation,'' \emph{IEEE
  Trans. Information Theory}, vol.~60, no.~10, pp. 6389--6398, Oct 2014.

\end{thebibliography}

%\end{spacing}
\end{document}